\pdfoutput=1
\documentclass[manuscript, nonacm]{acmart}
\AtBeginDocument{%
  \providecommand\BibTeX{{%
    \normalfont B\kern-0.5em{\scshape i\kern-0.25em b}\kern-0.8em\TeX}}}

\setcopyright{acmcopyright}
\copyrightyear{2018}
\acmYear{2018}
\acmDOI{10.1145/1122445.1122456}

\acmConference[Woodstock '18]{Woodstock '18: ACM Symposium on Neural
  Gaze Detection}{June 03--05, 2018}{Woodstock, NY}
\acmBooktitle{Woodstock '18: ACM Symposium on Neural Gaze Detection,
  June 03--05, 2018, Woodstock, NY}
\acmPrice{15.00}
\acmISBN{978-1-4503-XXXX-X/18/06}

\usepackage[utf8]{inputenc}
\usepackage{listings}
\usepackage{amsmath}
\usepackage{comment}
\usepackage[skins, breakable]{tcolorbox}
\usepackage{multirow}
\usepackage{amsthm}
\usepackage{amsfonts}
\usepackage{natbib}
\usepackage{graphicx}
\usepackage{appendix}
\usepackage{xspace}
\usepackage{color}
\usepackage{tikz}
\usepackage{pgfplots}
\usepackage{textcomp}
\usepackage{enumitem}
\usepackage{float}
\usepackage{hhline}
\usepackage{bm}
\usepackage{hyperref}

\newtheorem{lemma}{Lemma}
\newtheorem{theorem}{Theorem}
\newtheorem{definition}{Definition}

\newcommand{\fc} {\mathcal{C}}

\newcommand{\fs} {\mathcal{S}}
\newcommand{\ft} {\mathcal{T}}

\newcommand{\vote}{\texttt{vote}\xspace}

\newcommand{\latency} {good-case latency\xspace}

\newcommand{\timeout}{\texttt{timeout}\xspace}
\newcommand{\status}{\texttt{status}\xspace}

\newcommand{\vbbshort}{{psync-VBB}\xspace}

\newtcolorbox{mybox}[1][]{
enhanced,
colback=white,
boxsep=0pt,
#1
} 

\newcommand{\stitle}[1]{\vspace{0.5ex} \noindent\textsf{\textbf{#1}}}

\renewcommand{\paragraph}[1]{\smallskip\stitle{#1}}

\begin{document}

\title{Brief Note: Fast Authenticated Byzantine Consensus}
\titlenote{This is a complementary note of our previous paper~\cite{abraham2021goodcase} on the \latency of Byzantine broadcast.}
\author{Ittai Abraham}
\affiliation{%
  \institution{VMware Research}
  \country{Israel}
}
\email{iabraham@vmware.com}

\author{Kartik Nayak}
\affiliation{%
  \institution{Duke University}
  \country{USA}
}
\email{kartik@cs.duke.edu}

\author{Ling Ren}
\affiliation{%
  \institution{University of Illinois at Urbana-Champaign}
  \country{USA}
}
\email{renling@illinois.com}

\author{Zhuolun Xiang}
\affiliation{%
  \institution{University of Illinois at Urbana-Champaign}
  \country{USA}
}
\email{xiangzl@illinois.com}


\begin{abstract}
    Byzantine fault-tolerant (BFT) state machine replication (SMR) has been studied for over 30 years. Recently it has received more attention  due to its application in permissioned blockchain systems. A sequence of research efforts focuses on improving the commit latency of the SMR protocol in the common good case, including PBFT~\cite{castro1999practical} with $3$-round latency and $n\geq 3f+1$ and FaB~\cite{martin2006fast} with $2$-round latency and $n\geq 5f+1$.
    In this paper, we propose an authenticated protocol that solves $2$-round BFT SMR with only $n\geq 5f-1$ replicas, which refutes the optimal resiliency claim made in FaB for needing $n \geq 5f+1$ for $2$-round PBFT-style BFT protocols.
    For the special case when $f=1$, our protocol needs only $4$ replicas, and strictly improves PBFT by reducing the latency by one round (even when one backup is faulty).
\end{abstract}




\maketitle

\section{Introduction}

Byzantine fault-tolerant (BFT) state machine replication (SMR), which ensures all non-faulty replicas agree on the same sequence of client inputs to provide the client with the illusion of a single non-faulty server, is an important practical problem for building resilient distributed systems such as permissioned blockchain.
Most of the existing solutions to BFT SMR are leader-based, where a designated leader will drive consensus decisions for each view in the steady state, until it is replaced by the next leader via view-change due to malicious behavior or network partitions.
The design of BFT SMR usually focuses on optimizing the performance of the protocol under the common good case, when {\em an honest leader is in charge and the network is synchronous}. In particular, a sequence of research efforts aim at improving the latency of the good case~\cite{castro1999practical, kotla2007zyzzyva, martin2006fast, gueta2019sbft, hanke2018dfinity, chan2018pili, synchotstuff, abraham2020brief} for BFT SMR protocols.

This work also focuses on improving the good-case latency of the BFT SMR protocol under partial synchrony.
The most well-known BFT SMR protocol for partial synchrony is PBFT~\cite{castro1999practical}, which requires $3$ rounds of message exchange ($1$ round of proposing and $2$ rounds of voting) to commit a value in the good case, and has the optimal resilience of $n\geq 3f+1$~\cite{dwork1988consensus}.
For the good case when {\em the leader is honest but other replicas may be faulty} and the network is synchronous, FaB~\cite{martin2006fast} improves the latency to $2$ rounds ($1$ round of proposing and $1$ rounds of voting) with $n\geq 5f+1$ replicas.
Like PBFT, FaB only uses message authentication codes (MACs) instead of signatures in the steady state. 
The authors of the FaB paper claim that $n=5f+1$ is the best possible resilience for $2$-round BFT protocols.\footnote{Quote from Section 4 of FaB~\cite{martin2006fast}, ``Adding
signatures would reduce neither the number of communication steps nor the number of servers since FaB is already optimal in these two measures.''}

\paragraph{Summary of results.}
In this paper, we refute the above claim made in FaB~\cite{martin2006fast} by showing an authenticated BFT SMR protocol with $2$-round \latency that only requires $n\geq 5f-1$.
The BFT SMR protocol is an extension of our previous result for partially synchronous validated Byzantine broadcast~\cite{abraham2021goodcase}.
A special case of our SMR protocol is that, perhaps surprisingly, for the canonical example with $n=4$ and $f=1$ (note that $5f-1=3f+1$ in this case), we can design BFT SMR with 2-rounds in the good case.

\paragraph{Other Related Works.}
For the optimistic case when {\em all replicas are honest} and the network is synchronous,
Zyzzyva~\cite{kotla2007zyzzyva} and SBFT~\cite{gueta2019sbft} adds an optimistic commit path of $2$-round latency to PBFT with $n\geq 3f+1$ replicas.
Notice that the conditions for achieving the good case is much weaker than the optimistic case, since the good case  requires only the leader, instead of all replicas, to be honest. For example, with $n=4$ and $f=1$, our protocol has a  $2$-round latency even if one backup replica is malicious, while Zyzzyva or SBFT still requires 3-round latency in this case.

\section{Preliminaries}
\label{prelim}

The systems consists $n$ replicas numbered $1,2,...,n$. There exists at most $f$ Byzantine replicas with arbitrary behaviors controlled by an adversary. Rest of the replicas are called honest.
For simplicity of the presentation, we assume $n=5f-1$ in this paper, but the results apply to $n\geq 5f-1$.
The network model is the standard partial synchrony model~\cite{dwork1988consensus}, where the message delays are unbounded before an unknown Global Stable Time (GST). After GST, all messages between honest replicas will arrive within time $\Delta$.
The network channels are point-to-point, authenticated and reliable.
Unlike PBFT that can be implemented with Message Authentication Code (MAC) instead of digital signatures~\cite{castro1999practical}, our protocol relies on the PKI and digital signatures to detect equivocation.
We assume standard digital signatures and public-key infrastructure (PKI), and use $\langle m \rangle_i$ to denote a signed message $m$ by replica $i$.
In the paper, a message $m$ is valid if and only if $m$ is in the correct format and properly signed.
For simplicity, we assume the cryptographic primitives are ideal, to avoid the analysis of security parameters and negligible error probabilities.


        
        


In the problem of BFT SMR, clients send values to replicas, and the replicas provide the clients with the illusion of a single honest replica, by ensuring that all honest replicas agree on the same sequence of values. 

\begin{definition}[Byzantine Fault Tolerant State Machine Replication]
A Byzantine fault tolerant state machine replication protocol commits clients' values as a linearizable log akin to a single non-faulty server, and provides the following two guarantees.
\begin{itemize}[noitemsep,topsep=0pt]
    \item Safety. Honest replicas do not commit different values at the same log position. 
    \item Liveness. Each client value is eventually committed by all honest replicas.
    \item External Validity. 
    If an honest replica commits a value $v$, then $v$ is externally valid.
\end{itemize}
\end{definition}

Any client can send its value to at least $f+1$ replicas, and communicate with $f+1$ replicas to learn the committed value sequence proved by the corresponding commit certificate.
For most of the paper, we omit the client from the discussion and only focus on replicas.

We use the \latency metric~\cite{abraham2021goodcase} to measure the performance of our protocols, defined as follows.

\begin{definition}[Good-case Latency]
    The \latency of a BFT SMR protocol is the number of rounds needed for all honest replicas to commit, when the leader is honest and the network is synchronous.
\end{definition}

For instance, the classic PBFT~\cite{castro1999practical} has \latency of $3$ rounds with $n\geq 3f+1$ replicas, and FaB~\cite{martin2006fast}  has \latency of $2$ rounds with $n\geq 5f+1$ replicas.
We will propose protocols with \latency of $2$ rounds and $n\geq 5f-1$ for BFT SMR.

\section{$2$-round BFT Replication}
\label{sec:smr}
In this section, we propose an authenticated protocol with \latency of $2$ rounds that only needs $n\geq 5f-1$ replicas.
Following a recent line of work on chain-based BFT SMR~\cite{yin2019hotstuff, baudet2019state, synchotstuff, shrestha2020optimality} that commits a chain of blocks each containing a batch of clients' transactions, we use the following terminologies.

\paragraph{Block format, block extension and conflicting blocks.}
Clients' transactions (values) are batched into blocks, and the protocol outputs a chain of blocks $B_1,B_2,...,B_k,...$ where $B_k$ is the block at height $k$.
Each block $B_k$ has the following format $B_k=(h_{k-1}, k, txn)$ where $txn$ is a batch of new client transactions and $h_{k-1}=H(B_{k-1})$ is the hash digest of the previous block at height $k-1$.
We say that a block $B_l$ {\em extends} another block $B_k$, if $B_k$ is an ancestor of $B_l$ according to the hash chaining where $l\geq k$.
We define two blocks $B_l$ and $B'_{l'}$ to be {\em conflicting}, if they are not equal and do not extend on another. 
%
The block chaining simplifies the protocol in the sense that once a block is committed, its ancestors can also be committed.
    
Since the clients' transactions are batched into blocks, the BFT SMR protocol achieves {\em safety} if honest replicas always commit the same block $B_k$ for each height $k$, {\em liveness} if all honest replicas keep committing new blocks,
and {\em external validity} if all the committed blocks are externally valid.

\paragraph{Quorum certificate, timeout certificate, certificate ranking.}
    $\fc_w$ is a valid quorum certificate (QC) of view $w$ that {\em certifies} an externally valid block $B$ iff
    it consists of $\geq n-f=4f-1$ distinct signed \vote messages for block $B$ in the form of $\langle \vote, \langle B, w\rangle_{L_w} \rangle$ where $L_w$ is the leader of view $w$.  
    QCs and certified blocks are ranked first by the view numbers and then by the heights of the blocks, that is, QCs/blocks with higher views have higher ranks, and QCs/blocks with higher height have higher ranks if the view numbers are equal.

    $\ft_w$ is a valid timeout certificate (TC) of view $w$ that {\em locks} an externally valid block $B$ 
    iff 
    it consists of $\geq 4f-1$ signed \timeout messages in the form of $\langle \timeout, \langle B', w\rangle_{L_w} \rangle$ where $B'$ is some (possibly different) externally valid block or $\langle \timeout, \langle \bot , w\rangle \rangle$, and
    (1) 
    it contains $\geq 2f-1$ $\langle B', w \rangle_{L_w}$ where for each (possibly different) $B'$, $B$ equals or directly extends $B'$, and contains no block that conflicts $B$, or 
    (2) 
    it contains $\geq 2f$ $\langle B', w \rangle_{L_w}$ where for each (possibly different) $B'$, $B$ equals or directly extends $B'$, and no \timeout message from $L_{w}$.
    If multiple blocks satisfy the above conditions, we let $\ft_w$ lock the highest block.
    
    For example, let $B_1$ be the parent block of $B_2$. 
    If a TC contains $f$ signed \timeout messages for $B_1$, $f-1$ signed \timeout messages for $B_2$ and $2f$ signed \timeout messages for $\bot$, then it locks $B_2$ according to condition (1).
    If a TC contains $f$ signed \timeout messages for $B_1$, $f$ signed \timeout messages for $B_2$ and $2f-1$ signed \timeout messages for some conflicting block $B'$ and no \timeout message from the previous leader, then it locks $B_2$ according to condition (2).

    To bootstrap, we assume there is a certified genesis block $B_0$ of height $0$ that all honest replicas agree on when the protocol starts. The first leader of view $1$ will propose a block extending $B_0$, and every replica will vote for the block. For view-change, all honest replicas are assumed to have voted for $B_0$, and have the QC for $B_0$.

\subsection{Protocol $(5f-1)$-SMR}

\begin{figure}[h]
    \centering
    \begin{mybox}
    The protocol proceeds in view $w=1,2,...$, each with a leader $L_{w}$. 
    Each replica locally maintains the highest \timeout certificate $\ft_{high}$.
    The honest replicas will ignore any message for a block that is not externally valid.
    
\textbf{Steady State Protocol for Replica $i$}

Let $w$ be the current view number and replica $L_w$ be the current leader.
    
\begin{enumerate}
    
    \item\label{smr:step:propose} \textbf{Propose.} 
    The leader $L_w$ multicasts $\langle \texttt{propose}, \langle B_k, w\rangle_{L_w}, \fc, \fs \rangle_{L_w}$.
    If $B_k$ is not the first block proposed in view $w$, then $B_k$ is a new externally valid block extending the last block $B_{k-1}$ proposed by $L_w$, $\fc$ is the QC that certifies $B_{k-1}$, and $\fs=\emptyset$;
    Otherwise $B_k, \fc, \fs$ are specified in the {\em Status} step.
    
    \item\label{smr:step:vote} \textbf{Vote.} 
    Upon receiving a signed proposal $\langle \texttt{propose}, \langle B_k, w\rangle_{L_w}, \fc, \fs \rangle_{L_w}$ from the leader $L_w$,
    \begin{itemize}[itemsep=0pt,topsep=0pt]
        \item if $B_k$ is the first proposed block in view $w$, check if
        (1) $\fs$ is a valid TC of view $w-1$ that locks $B_k$ and $\fc$ is a valid QC of the parent block of $B_k$, 
        or (2) $\fs$ contains $4f-1$ valid \status messages of view $w-1$, $B_k$ is locked by the highest TC in $\fs$, and $\fc$ is a valid QC of the parent block of $B_k$;
        \item otherwise, check if $B_k$ extends the highest certified block known to the replica.
    \end{itemize}
    If one of the above condition is true, and the replica hasn't voted for any other height-$k$ block,
    multicast a \vote message in the form of $\langle \vote, \langle B_k, w\rangle_{L_w} \rangle_i$.

    \item\label{smr:step:commit} \textbf{Commit.}
    When receiving $4f-1$ signed \vote messages of view $w$ for the same block $B$, form a QC, forward the QC to all other replicas, and commit $B$ with all its ancestors blocks.

\end{enumerate}

\textbf{View-change Protocol for Replica $i$}
\begin{enumerate}[itemsep=0pt,topsep=0pt]
    
    \item\label{smr:step:timeout} \textbf{Timeout.}
    If less than $p$ valid blocks are committed within $(2p+2)\Delta$ time after entering view $w$, timeout view $w$ and stop voting for view $w$, and multicast $\langle \timeout, \langle B ,w\rangle_{L_w} \rangle_i$ where $B$ is the highest block voted in view $w$ (multicast $\langle \timeout, \langle \bot ,w\rangle \rangle_i$ if not voted for any).
    
    \item\label{smr:step:newview} \textbf{New View.} Upon receiving $4f-1$ valid \timeout messages of view $w-1$ that contains no conflicting blocks signed by $L_{w-1}$, or $4f-1$ valid \timeout messages from replicas other than $L_{w-1}$, perform the following:
    Forward these \timeout messages. 
    If the \timeout messages can form a \timeout certificate $\ft_{w-1}$ that locks a block, then update  $\ft_{high}=\ft_{w-1}$.
    Timeout view $w-1$ if haven't, and enter view $w$.
    Send a \status message in the form of $\langle \status, w-1, \fc, \ft_{high} \rangle_i $ to the leader $L_w$, where $\fc$ is the QC of the parent block of the block that $\ft_{high}$ locks.
    
    \item\label{smr:step:status} \textbf{Status.}
    After entering view $w$ and receiving $4f-1$ valid \status messages of view $w-1$,
    the leader $L_w$ sets the first new proposal block $B$, QC of the parent block $\fc$, and a proof $\fs$ as follows.
    \begin{itemize}[itemsep=0pt,topsep=0pt]
        
        \item If any valid TC $\ft$ of view $w-1$ locks a block $B'$, set $\fs=\ft$, $B=B'$ and $\fc$ to be the QC of the parent block of $B'$.
        
        \item Otherwise, set $\fs$ to be the set of $4f-1$ valid \status messages of view $w-1$ received, set $B$ to be the block locked by the highest $\ft$ in $\fs$, and set $\fc$ to be the QC of the parent block of $B$.
    \end{itemize}
\end{enumerate}

    \end{mybox}
    \caption{$(5f-1)$-SMR Protocol with \latency of $2$ rounds}
    \label{fig:smr}
\end{figure}

\paragraph{Protocol Description.}
Now, we present the protocol in Figure~\ref{fig:smr}, and briefly describe the protocol below. 
Each leader can keep proposing blocks until it is replaced by the next leader.
All committed blocks form a chain linked by hash digest and QC, where QC consists $n-f=4f-1$ votes of the parent block.
To propose the first block in the new view, the leader will wait for $4f-1$ valid \status messages.
If the \status messages contain a valid TC of the previous view, the leader proposes the block locked by this TC; otherwise, the leader proposes the block locked by the highest TC among \status messages.
Once the first block is voted by $n-f=4f-1$ replicas, the leader can propose the next block extending the first block together with the QC of the first block.
Then, whenever a block gets certified by a QC, it can be committed, and the leader can propose the next block extending this block.
When not enough progress is made in the current view (less than $p$ blocks committed within $(2p+2)\Delta$ time), the replica stops voting for the current view and multicast a \timeout message for the current view.
The \timeout message contains the highest block voted by the replica in the current view (if not voted then contains $\bot$).
The \timeout messages serve a similar purpose of the \timeout messages in the protocol $(5f-1)$-\vbbshort of our previous paper~\cite{abraham2021goodcase}, that is to form a TC that can lock a block for the next leader to propose.
The guarantee is that, if a block $B$ is committed at any honest replica, then any honest replica will have a TC that locks $B$ during the view-change, and no valid TC can lock on other conflicting blocks.
The replicas will enter the next view after receiving $n-f=4f-1$ \timeout messages, and send a \status message to the new leader containing its highest TC. Sending the highest TC ensures that even if the previous view has no progress, the highest TC in the \status messages will lock the highest block committed in the earlier views.

\subsection{Proof of Correctness}

\begin{lemma}\label{lem:smr:1}
    If an honest replica directly commits a block $B_k$ in view $w$, then any certified block $B_{k'}$ of view $w$ and height $k'\geq k$ must equal or extend $B_k$.
\end{lemma}

\begin{proof}
    Suppose $k'=k$. Since any committed or certified block need $4f-1$ votes, if two different blocks of height $k$ are both certified, then by quorum intersection, there should exist at least $(4f-1)+(4f-1)-(5f-1)=3f-1>f$ Byzantine parties, which is a contradiction. 
    Now suppose $k'>k$. Since $B_{k'}$ does not extend $B_k$, there must exist a certified block $B_k'$ that $B_{k'}$ extends since honest parties only vote for blocks that extend certified blocks, and $B_k'$ and $B_k$ conflicts each other. However, such certified $B_k'$ cannot exist by earlier argument, and thus 
    any certified block $B_{k'}$ of view $w$ and height $k'\geq k$ must equal or extend $B_k$.
\end{proof}

\begin{lemma}\label{lem:smr:2}
    If block $B_k$ is the highest block certified in view $w$, then no valid TC of view $w$ can lock any block that conflicts $B_k$, and any honest replica that enters view $w+1$ has a valid TC of view $w$ that locks $B_k$ or some $B_{k+1}$ directly extending $B_k$.
\end{lemma}

\begin{proof}
    By default, any message discussed below is of view $w$.
    Since $B_k$ is certified in view $w$, at least $3f-1$ honest replicas voted for $B_k$. Since $B_k$ is the highest certified block in view $w$, no honest replica voted for any block $B_{k+2}$ (but they may vote for some $B_{k+1}$ that extends $B_k$).
    Hence, at least $3f-1$ honest replicas include $B_k$ or some $B_{k+1}$ that extends $B_k$ in their \timeout messages, but not any $B'$ that conflicts $B_k$.
    
    By definition, a valid TC $\ft$ that locks $B$
    iff 
    it consists of $\geq 4f-1$ signed \timeout messages in the form of $\langle \timeout, \langle B , w\rangle_{L_w} \rangle$ where $B$ is some externally valid block or $\langle \timeout, \langle \bot , w\rangle \rangle$, and
    (1) 
    it contains $\geq 2f-1$ $\langle B', w \rangle_{L_w}$ where for each (possibly different) $B'$, $B$ equals or directly extends $B'$, and contains no block that conflicts $B$, or 
    (2) 
    it contains $\geq 2f$ $\langle B', w \rangle_{L_w}$ where for each (possibly different) $B'$, $B$ equals or directly extends $B'$, and no \timeout message from $L_{w}$.
    
    First we prove that no valid TC of view $w$ can lock a block $B'$ that conflicts $B_k$.
    Suppose that there exists a valid TC $\ft$ of view $w$ that locks a block $B'$ that conflicts $B_k$. 
    Condition (1) cannot be true: 
    Since $3f-1$ honest replicas include blocks that conflict $B'$ in their \timeout messages, $\ft$ cannot include these signatures, which implies there need to be at least $3f-1+4f-1=7f-2>n$ replicas and is impossible.
    Condition (2) also cannot be true: 
    If $L_{w}$ is honest, then no conflicting block can be signed by $L_{w}$.
    If $L_{w}$ is Byzantine, then $\ft$ contains at most $f-1$ signatures from Byzantine replicas since the leader $L_w$ is excluded. Since at most $(4f-1)-(3f-1)=f$ honest replicas include $B'$ in \timeout, at most $f+f-1=2f-1$ signatures can be on $B'$.
    Therefore, there exists no valid TC $\ft$ of view $w$ that locks any block $B'$ that conflicts $B_k$.

    Now we prove that any honest replica that enters view $w+1$ has a valid TC of view $w$ that locks $B_k$ or some $B_{k+1}$ directly extending $B_k$.
    Consider any honest replica that enters view $w+1$. According to Step~\ref{smr:step:newview} of view-change, the replica receives either $4f-1$ valid \timeout messages of view $w$ that contain no conflicting blocks signed by $L_w$, or $4f-1$ valid \timeout messages from replicas other than $L_{w}$.
    There are two cases.
    If the received $4f-1$ \timeout messages contain no conflicting blocks signed by $L_w$, since at least $4f-1-2f=2f-1$ of them contain $B_k$ or some $B_{k+1}$ directly extending $B_k$, condition (1) for locking is satisfied.
    If the received $4f-1$ \timeout messages are from replicas other than $L_w$, if $L_w$ is honest then condition (1) also holds;
    Otherwise, if $L_w$ is Byzantine, then the set of $4f-1$ \timeout messages contains at most $f-1$ \timeout messages from Byzantine replicas since the leader $L_w$ is excluded. Since at most $(4f-1)-(3f-1)=f$ honest replicas may include conflicting blocks in \timeout, the set of $4f-1$ \timeout messages includes at least $(4f-1)-f-(f-1)=2f$ \timeout that contain $B_k$ or some $B_{k+1}$ directly extending $B_k$. Then condition also holds.
    Hence, any honest replica that enters view $w+1$ has a valid TC of view $w$ that locks $B_k$ or some $B_{k+1}$ directly extending $B_k$.
\end{proof}

\begin{lemma}\label{lem:smr:3}
    If an honest replica directly commits a block $B$, then any certified block that ranks no lower than $B$ must equal or extend $B$.
\end{lemma}

\begin{proof}
    Suppose an honest party $h$ directly commits a block $B$ in view $w$.
    
    First, for any certified block of view $w$, the lemma is true by Lemma~\ref{lem:smr:1}.
    
    Now we prove the lemma for view $>w$ by first proving that in view $>w$ (1) the honest replicas only vote for blocks that equal or extend $B$, and (2) any valid TC of view $\geq w$ only locks blocks that equal or extend $B$.
    We prove by induction on the view number.
    Consider the base case of view $w+1$. By Lemma~\ref{lem:smr:1} and \ref{lem:smr:2}, any valid TC of view $w$ can only lock a block that equals or extends $B$, and any honest replica that enters view $w+1$ has a valid TC
    of view $w$ that locks $B$ or some block directly extending $B$. 
    Therefore, any set of $4f-1$ valid \status messages of view $w$ must contain a valid TC of view $w$ that locks a block that equals or extends $B$, and honest replicas will only for blocks that equal or extends $B$ according to Step~\ref{smr:step:vote} of the steady state protocol.
    Assume the induction hypothesis that any honest replica only votes for a block that equals or extends $B$ in view $w+1,...,k-1$, and any valid TC of view $w,...,k-2$ only locks blocks that equal or extend $B$. 
    Then, any valid TC of view $k-1$ can only lock a block that equals or extends $B$ by a proof similar to that of Lemma~\ref{lem:smr:2}.
    Since any valid TC of view $w,...,k-1$ can only lock a block that equals or extends $B$, and any honest replica that enters view $w+1$ has a valid TC
    of view $w$ that locks $B$, any honest replica in view $k$ will only vote for blocks that equal or extends $B$ according to Step~\ref{smr:step:vote} of the steady state protocol.
    Therefore, the claim is true by induction, which also implies that any certified block that ranks no lower than $B$ must equal or extend $B$.
\end{proof}

\begin{theorem}[Safety]\label{thm:smr:safety}
    Honest replicas always commit the same block $B_k$ for each height $k$.
\end{theorem}

\begin{proof}
    Suppose two blocks $B_k$ and $B_k'$ are committed at height $k$ at any two honest replicas. Suppose $B_k$ is committed due to $B_l$ being directly committed in view $v$, and $B_k'$ is committed due to $B_{l'}'$ being directly committed in view $v'$.
    Without loss of generality, suppose $v\leq v'$, and for $v=v'$, further assume that $l\leq l'$. 
    Since $B_l$ is directly committed and $B_{l'}$ is certified and ranks no lower than $\fc_v(B_l)$,
    by Lemma \ref{lem:smr:3}, $B_{l'}$ must equal or extend $B_l$. Thus, $B_k'=B_k$.
\end{proof}

\begin{theorem}[Liveness]\label{thm:smr:liveness}
    All honest replicas keep committing new blocks.
\end{theorem}

\begin{proof}
    After GST, when the leader is honest, all honest replicas will keep committing new blocks and no honest replica will send \timeout message.
    Since the leader is honest, it will set the first proposal block according to Step~\ref{smr:step:status} of view-change, and all honest replicas will vote for the block according to Step~\ref{smr:step:vote} of the steady state protocol. For later blocks, since the leader will extend the last proposed block, all honest replicas will vote according to Step~\ref{smr:step:vote} of the steady state protocol. 
    The time window $(2p+2)\Delta$ is sufficient for an honest replica to commit $p$ blocks, since the leader may enter the view at most $\Delta$ time later, then wait for $\Delta$ to receive the \status messages. After that, each proposed block takes $2\Delta$ to be committed, since after $\Delta$ time the block is received at all honest replicas, and after another $\Delta$ time, all the votes are received by all honest replicas, leading to the commit.
    Hence, no honest replica will send \timeout message, and the leader will not be replaced.
    
    Otherwise, if the network is asynchronous or the leader is Byzantine, the honest replicas may not commit enough blocks and thus send \timeout messages. When $f+1$ honest replicas send \timeout messages and thus stop voting for any block in this view, no new blocks can be certified and all $4f-1$ honest replicas will eventually send \timeout messages and enter the next view.
    Eventually, after GST and an honest leader is elected, all honest replicas will keep committing new blocks.
\end{proof}

\begin{theorem}[External Validity]
    The block committed by any honest replicas is externally valid.
\end{theorem}
\begin{proof}
    Since any honest replica only votes for blocks that are externally valid, the claim is trivially true.
\end{proof}

\begin{theorem}[Good-case Latency]
    When the network is synchronous and the leader is honest, the proposal of the leader will be committed within $2$ rounds.
\end{theorem}
\begin{proof}
    Since the leader is honest, it proposes the same block to all honest replicas. Then all honest replicas will vote for the block as proved in Theorem~\ref{thm:smr:liveness}, and therefore commit within $2$ rounds of message exchanges after receiving all votes from the honest replicas.
\end{proof}

\section{Conclusion}

In this paper, we extend the result for partially synchronous validated Byzantine broadcast from our previous paper~\cite{abraham2021goodcase}, to obtain a chain-based BFT SMR protocol with $n\geq 5f-1$ and commit latency of $2$ rounds in the good case.

\section*{Acknowledgement}
We would like to thank Irene Isaac and Tzu-Bin Yan for pointing out a subtle issue for the certificate definition.

\bibliographystyle{ACM-Reference-Format}
\bibliography{references}

\appendix

\end{document}